\documentclass[11pt]{amsart}

\usepackage{amssymb,latexsym}
\usepackage{graphicx,epsfig,color}

\newtheorem{thm}{Theorem}[section]

\newtheorem{cor}[thm]{Corollary}
\newtheorem{lem}[thm]{Lemma}
\newtheorem{rem}[thm]{Remark}

\newtheorem{defn}[thm]{Definition}

\numberwithin{equation}{section}

\date{}

\begin{document}

\author[Mukhiddin I. Muminov, Tulkin H. Rasulov, Nargiza A. Tosheva]{Mukhiddin I. Muminov$^1$, Tulkin H. Rasulov$^2$, Nargiza A. Tosheva$^2$}
\title[Analysis of the discrete spectrum of the family of $3 \times 3$ operator
matrices]{Analysis of the discrete spectrum of the family of\\ $3 \times 3$
operator matrices} \maketitle

\begin{center}
{\small $^1$Faculty of Scinces, Universiti Teknologi Malaysia (UTM)\\
 81310 Skudai, Johor Bahru, Malaysia\\
 E-mail: mmuminov@mail.ru

$^2$Faculty of Physics and Mathematics, Bukhara State University\\
M. Ikbol str. 11, 200100 Bukhara, Uzbekistan\\
E-mail: rth@mail.ru, nargiza$_{-}$n@mail.ru}
\end{center}

\begin{abstract}
We consider the family of $3 \times 3$ operator matrices ${\bf
H}(K),$ $K \in {\Bbb T}^3:=(-\pi; \pi]^3$ associated with the
lattice systems describing two identical bosons and one particle,
another nature in interactions, without conservation of the number
of particles. We find a finite set $\Lambda \subset {\Bbb T}^3$ to
prove the existence of infinitely many eigenvalues of ${\bf H}(K)$
for all $K \in \Lambda$ when the associated Friedrichs model has a
zero energy resonance. It is found that for every $K \in \Lambda,$
the number $N(K, z)$ of eigenvalues of ${\bf H}(K)$ lying on the
left of $z,$ $z<0,$ satisfies the asymptotic relation
$\lim\limits_{z \to -0} N(K, z) |\log|z||^{-1}={\mathcal U}_0$
with $0<{\mathcal U}_0<\infty,$ independently on the cardinality
of $\Lambda.$ Moreover, we prove that for any $K \in \Lambda$ the
operator ${\bf H}(K)$ has a finite number of negative eigenvalues
if the associated Friedrichs model has a zero eigenvalue or a zero
is the regular type point for positive definite Friedrichs model.
\end{abstract}



\section{Introduction}

The main objective of the present paper is to establish the
finiteness or infiniteness of the number of eigenvalues for a
family of $3 \times 3$ operator matrices ${\bf H}(K),$ $K \in
{\Bbb T}^3:=(-\pi; \pi]^3$ and especially the asymptotics for the
number of infinitely many eigenvalues (Efimov's effect case).
These operator matrices are associated with the lattice systems
describing two identical bosons and one particle, another nature
in interactions, without conservation of the number of particles.

The Efimov effect is one of the most remarkable results in the
spectral analysis for continuous three-particle Schr\"{o}dinger
operators: if none of the three two-particle Schr\"{o}dinger
operators (corresponding to the two-particle subsystems) has
negative eigenvalues but at least two of them have zero energy
resonance, then the three-particle Schr\"{o}dinger operator has
infinitely many negative eigenvalues accumulating at zero.

For the first time the Efimov effect has been discussed in
\cite{Efim}. Then this problem has been studied on a physical
level of rigor in \cite{AHW, Amad-Nob}. A rigorous mathematical
proof of the existence of Efimov's effect was originally carried
out in \cite{Yaf} and then many works devoted to this subject, see
for example \cite{Dell-Fig-Teta,Ovch-Sig,Sob,Tam-1,Tam-2}. The
main result obtained by Sobolev \cite{Sob} (see also \cite{Tam-2})
is an asymptotics of the form ${\mathcal U}_0 |\log|z||$ for the
number $N(z)$ of eigenvalues on the left of $z,$ $z<0,$ where the
coefficient ${\mathcal U}_0$ does not depend on the two-particle
potentials $v_\alpha$ and is a positive function of the ratios
$m_1/m_2$ and $m_2/m_3$ of the masses of the three particles.

In a system of three-particles on three-dimensional lattices, due
to the fact that the discrete analogue of the Laplacian or its
generalizations are not rationally invariant, the Hamiltonian of a
system does not separate into two parts, one relating to the
center-of-mass motion and the other one to the internal degrees of
freedom. In particular, in this case the Efimov effect exists only
for the zero value of the three-particle quasi-momentum $K \in
{\Bbb T}^3$ (see \cite{AL, ALM, LM}). An asymptotics analogous to
\cite{Sob,Tam-2} was obtained in \cite{AL, ALM} for the number of
eigenvalues.

In all above mentioned papers devoted to the Efimov effect, the
systems where the number of quasi-particles is fixed have been
considered. In the theory of solid-state physics \cite{Mog},
quantum field theory \cite{Frid}, statistical physics
\cite{Mal-Min, MS}, fluid mechanics \cite{Cha61},
magnetohydrodynamics \cite{Lif89} and quantum mechanics
\cite{Tha92} some important problems arise where the number of
quasi-particles is finite, but not fixed. In \cite{SSZ} geometric
and commutator techniques have been developed in order to find the
location of the spectrum and to prove absence of singular
continuous spectrum for Hamiltonians without conservation of the
particle number.

In the present paper we consider the family of $3 \times 3$
operator matrices ${\bf H}(K),$ $K \in {\Bbb T}^3$ associated with
the lattice systems describing two identical bosons and one
particle, another nature in interactions, without conservation of
the number of particles. This operator acts in the direct sum of
zero-, one- and two-particle subspaces of the bosonic Fock space
and it plays a key role for the study of the energy operator of 
the spin-boson Hamiltonian with two bosons on the torus \cite{MNR2015, Ras2016}.
We discuss the case where the dispersion function has
form $\varepsilon(p)=\sum\limits_{i=1}^3 (1-\cos(n p^{(i)}))$ with
$n>1.$ We denote by $\Lambda$ the set of points ${\Bbb T}^3$ where
the function $\varepsilon(\cdot)$ takes its (global) minimum.
Under some smoothness assumptions on the parameters of a family of
Friedrichs models ${\bf h}(k),$ $k \in {\Bbb T}^3,$ we obtain the
following results:

(i) We describe the location of the essential spectrum
$\sigma_{\rm ess}({\bf H}(K))$ of ${\bf H}(K),$ $K \in {\Bbb T}^3$
via the spectrum of ${\bf h}(k),$ $k \in {\Bbb T}^3;$

(ii) We prove that for all $K \in \Lambda$ the ${\bf H}(K)$ has
infinitely many negative eigenvalues accumulating at zero, if the
operator ${\bf h}({\bf 0}),$ ${\bf 0}=(0,0,0)$ has a zero energy
resonance (Efimov's effect). Moreover, for any $K \in \Lambda$ we
establish the asymptotics $N(K; z) \sim {\mathcal U}_0 |\log|z||$
with $0<{\mathcal U}_0<\infty$ for the number $N(K; z)$ of
eigenvalues of ${\bf H}(K)$ lying on the left of $z,$ $z<\min
\sigma_{\rm ess}({\bf H}(K))=0;$

(iii) We prove the finiteness of negative eigenvalues of ${\bf
H}(K)$ for $K \in \Lambda,$ if the operator ${\bf h}({\bf 0})$ has
a zero eigenvalue or a zero is the regular type point for ${\bf
h}({\bf 0})$ with ${\bf h}({\bf 0}) \geq 0.$

We remark that for the Friedrichs model ${\bf h}({\bf 0})$ the
presence of a zero energy resonance (consequently the existence of
the Efimov effect for ${\bf H}(K),$ $K \in \Lambda$) is due to the
annihilation and creation operators.

We point out that the operator ${\bf H}(K)$ has been considered
before in \cite{ALR, ALR1, LR} for $K=0$ and $n=1,$ and 
in \cite{Ras2011} for $K=0$ with arbitrary $n,$ where proven
the existence of Efimov's effect. Similar asymptotics for the
number of eigenvalues was obtained in \cite{ALR}.
We recall that the main results (without proofs) of this paper has been
announced in \cite{MumRas}. This paper is devoted to the detailed proof of
these results with respect to the number of eigenvalues.
The result related with the essential spectrum of ${\bf H}(K)$
was discussed in \cite{RasTosh2019}.

It surprising that in the assertion (ii) the asymptotics for $N(K;
z)$ is the same for all $K \in \Lambda$ and is stable with respect
to the number $n.$ Recall that in all papers devoted to Efimov's
effect for lattice systems the existence of this effect have been
proved only for zero value of the quasi-momentum ($K=0$) and for
the case $n=1,$ or only for zero value of the quasi-momentum with 
arbitrary $n$.

The organization of the present paper is as follows. Section 1 is
an introduction to the whole work. In Section 2, the operator
matrices ${\bf H}(K),$ $K \in {\Bbb T}^3$ are described as the
family of bounded self-adjoint operators in the direct sum of
zero-, one- and two-particle subspaces of the bosonic Fock space
and the main results are formulated. In Section 3, we discuss some
results concerning threshold analysis of the Friedrichs model
${\bf h}(k),$ $k \in {\Bbb T}^3.$ In Section 4 we give a
modification of the Birman-Schwinger principle for ${\bf H}(K),$
$K \in {\Bbb T}^3.$ Section 5 we establish the finiteness of the
number of eigenvalues of the operator ${\bf H}(K),$ $K \in
\Lambda.$ In section 6 we obtain the asymptotic formula for the
number of negative eigenvalues of ${\bf H}(K),$ $K \in \Lambda.$

We adopt the following conventions throughout the present paper.
Let ${\Bbb T}^3$ be the three-dimensional torus, the cube
$(-\pi,\pi]^3$ with appropriately identified sides equipped with
its Haar measure. Denote by $\sigma(\cdot),$ $\sigma_{\rm
ess}(\cdot)$ and $\sigma_{\rm disc}(\cdot),$ respectively, the
spectrum, the essential spectrum, and the discrete spectrum of a
bounded self-adjoint operator. In what follows we deal with the
operators in various spaces of vector-valued functions. They will
be denoted by bold letters and will be written in the matrix form.

\section{Family of $3 \times 3$ operator matrices and main results}

Let ${\Bbb C}$ be the field of complex numbers, $L_2({\Bbb T}^3)$
be the Hilbert space of square integrable (complex) functions
defined on ${\Bbb T}^3$ and $ L_2^{\rm s} (({\Bbb T}^3)^2)$ be the
Hilbert space of square integrable (complex) symmetric functions
defined on $({\Bbb T}^3)^2.$ Denote by ${\mathcal H}$ the direct
sum of spaces ${\mathcal H}_1={\Bbb C},$ ${\mathcal H}_1=L_2({\Bbb
T}^3)$ and ${\mathcal H}_2=L_2^{\rm s}(({\Bbb T}^3)^2),$ that is,
${\mathcal H}={\mathcal H}_0 \oplus {\mathcal H}_1 \oplus
{\mathcal H}_2.$ The spaces ${\mathcal H}_0,$ ${\mathcal H}_1$ and
${\mathcal H}_2$ are called zero-, one- and two-particle subspaces
of a bosonic Fock space ${\mathcal F}_{\rm s}(L_2({\Bbb T}^3))$
over $L_2({\Bbb T}^3),$ respectively. It is well-known that
if ${\mathcal A}$ is a bounded linear in a Hilbert space
${\mathcal H}$ and a decomposition ${\mathcal H}={\mathcal H}_0 \oplus {\mathcal H}_1 \oplus
{\mathcal H}_2$ into three Hilbert spaces ${\mathcal H}_0,$ ${\mathcal H}_1,$
${\mathcal H}_2$ is given, then ${\mathcal A}$ always admits \cite{Jeribi, Tre08} a block
operator matrix representation
$$
{\mathcal A}=\left( \begin{array}{ccc}
A_{00} & A_{01} & A_{02}\\
A_{10} & A_{11} & A_{12}\\
A_{20} & A_{21} & A_{22}\\
\end{array}
\right)
$$
with linear operators $A_{ij}: {\mathcal H}_j \to {\mathcal H}_i,$ $i,j=0,1,2.$

Let us consider the following family of $3 \times 3$ operator
matrices ${\bf H}(K),$ $K \in {\Bbb T}^3$ acting in the Hilbert
space ${\mathcal H}$ as
$$
{\bf H}(K):=\left( \begin{array}{ccc}
H_{00}(K) & H_{01} & 0\\
H_{01}^* & H_{11}(K) & H_{12}\\
0 & H_{12}^* & H_{22}(K)\\
\end{array}
\right)
$$
with the entries
\begin{align*}
& H_{00}(K)f_0=w_0(K)f_0, \quad H_{01}f_1=\int_{{\Bbb T}^3}
v_0(t)f_1(t)dt,\\ 
& (H_{11}(K)f_1)(p)=w_1(K; p)f_1(p),\quad
(H_{12}f_2)(p)= \int_{{\Bbb T}^3} v_1(t) f_2(p,t)dt,\\
& (H_{22}(K)f_2)(p,q)=w_2(K;p,q)f_2(p,q),
\end{align*}
where $H_{ij}^*$ $(i<j)$ denotes the adjoint operator to $H_{ij}$
and $f_i \in {\mathcal H}_i,$ $i=0,1,2.$

Here $w_0(\cdot)$ and $v_i(\cdot),$ $i=0,1$ are real-valued
bounded functions on ${\Bbb T}^3,$ the functions $w_1(\cdot;
\cdot)$ and $w_2(\cdot; \cdot, \cdot)$ are defined by the
equalities
\begin{align*}
& w_1(K; p):=l_1\varepsilon(p)+l_2\varepsilon(K-p)+1;\\
& w_2(K; p,q):=l_1\varepsilon(p)+l_1\varepsilon(q)+l_2\varepsilon(K-p-q),
\end{align*}
respectively, with $l_1,l_2>0$ and
$$
\varepsilon(q):=\sum_{i=1}^3(1-\cos(n q^{(i)})), \quad
q=(q^{(1)},q^{(2)},q^{(3)}) \in {\Bbb T}^3, \quad n \in {\Bbb N}.
$$

Under these assumptions the operator ${\mathbf H}(K)$ is bounded
and self-adjoint.

We remark that the operators $H_{01}$ and $H_{12}$ resp.
$H_{01}^*$ and $H_{12}^*$ are called annihilation resp. creation
operators, respectively. In this paper we consider the case, where
the number of annihilations and creations of the particles of the
considering system is equal to 1. It means that $H_{ij}\equiv 0$
for all $|i-j|>1.$

To study the spectral properties of the operator ${\bf H}(K)$ we
introduce a family of bounded self-adjoint operators (Friedrichs
models) ${\mathbf h}(k),$ $k\in {\Bbb T}^3,$ which acts in
${\mathcal H}_0 \oplus {\mathcal H}_1$ as
$$
{\mathbf h}(k)=\left( \begin{array}{cc}
h_{00}(k) & h_{01}\\
h_{01}^* & h_{11}(k)\\
\end{array}
\right),
$$
where
\begin{align*}
& h_{00}(k)f_0=(l_2 \varepsilon(k)+1) f_0,\quad
h_{01}f_1=\frac{1}{\sqrt{2}} \int_{{\Bbb T}^3} v_1(t)
f_1(t)dt,\\
& (h_{11}(k)f_1)(q)=E_k(q)f_1(q), \quad
E_k(q):=l_1\varepsilon(q)+l_2\varepsilon(k-q).
\end{align*}

The following theorem \cite{ALR, ALR1, LRmz, RasTosh2019} describes the location of
the essential spectrum of the operator ${\bf H}(K)$ by the
spectrum of the family ${\mathbf h}(k)$ of Friedrichs models.

\begin{thm}\label{THM 2.1} For the essential spectrum of ${\bf H}(K)$
the equality
\begin{equation}\label{ess of bfH}
\sigma_{\rm ess}({\bf H}(K))=\bigcup\limits_{p\in {{\Bbb
T}^3}}\{\sigma_{\rm disc}({\mathbf h}(K-p))+l_1 \varepsilon(p)\}
\cup [m_K; M_K]
\end{equation}
holds, where the numbers $m_K$ and $M_K$ are defined by
$$
m_K:= \min\limits_{p,q\in {\Bbb T}^3} w_2(K;p,q)\quad \mbox{and}
\quad M_K:= \max\limits_{p,q\in {\Bbb T}^3} w_2(K;p,q).
$$
\end{thm}

Let $\Lambda$ a subset of ${\Bbb T}^3$ given by
$$
\Lambda:= \left\{ (p^{(1)},p^{(2)},p^{(3)}): p^{(i)} \in \left\{0,
\pm \frac{2}{n} \pi; \pm \frac{4}{n} \pi; \dots; \pm \frac{n'}{n}
\pi \right\} \cup \Pi_n, \,\, i=1,2,3 \right\},
$$
where
$$
n':=\left \lbrace
\begin{array}{ll}
n-2,\,\, \mbox{if}\,\, n \,\, \mbox{is even} \\
n-1,\,\, \mbox{if}\,\, n \,\, \mbox{is odd}
\end{array} \right. \quad \mbox{and} \quad
\Pi_n:=\left \lbrace
\begin{array}{ll}
\{\pi\},\,\, \mbox{if}\,\, n \,\, \mbox{is even} \\
\,\,\emptyset,\,\quad \mbox{if}\,\, n \,\, \mbox{is odd}
\end{array} \right.
$$

Direct calculation shows that the cardinality of $\Lambda$ is
equal to $n^3.$ It is easy to check that for any $K \in \Lambda$
the function $w_2(K;\cdot,\cdot)$ has non-degenerate zero minimum
at the points  of $\Lambda \times \Lambda,$ that is, $m_K=0$ for
$K \in \Lambda.$

The following assumption we be needed throughout the paper: the
function $v_1(\cdot)$ is either even or odd function on each
variable and there exist all second order continuous partial
derivatives of $v_1(\cdot)$ on ${\Bbb T}^3.$

Since ${\bf 0}=(0,0,0)\in \Lambda$ the definition of the functions
$w_1(\cdot; \cdot)$ and $w_2(\cdot;\cdot,\cdot)$ implies the
identity ${\bf h}({\bf 0}) \equiv {\bf h}(k)$ for all $k \in
\Lambda.$

Let us denote by $C({\Bbb T}^3)$ and $L_1({\Bbb T}^3)$ the Banach
spaces of continuous and integrable functions on ${\Bbb T}^3,$
respectively.

\begin{defn}\label{defn of resonance}
The operator ${\bf h}({\bf 0})$ is said to have a zero energy
resonance, if the number $1$ is an eigenvalue of the integral
operator given by
$$
(G \psi)(q)=\frac{v_1(q)}{2(l_1+l_2)} \int_{{\Bbb T}^3}
\frac{v_1(t)\psi(t)}{\varepsilon(t)}dt, \quad \psi \in C({\Bbb
T}^3)
$$
and at least one $($up to a normalization constant$)$ of the
associated eigenfunctions $\psi$ satisfies the condition $\psi(p')
\neq 0$ for some $p' \in \Lambda.$ If the number $1$ is not an
eigenvalue of the operator $G,$ then we say that $z=0$ is a
regular type point for the operator ${\bf h}({\bf 0}).$
\end{defn}

We notice that in Definition $\ref{defn of resonance}$ the
requirement of the existence of the eigenvalue $1$ of $G$
corresponds to the existence of a solution of ${\bf h}({\bf
0})f=0$ and the condition $\psi(p') \neq 0$ for some $p' \in
\Lambda$ implies that the solution $f$ of this equation does not
belong to ${\mathcal H}_0 \oplus {\mathcal H}_1.$ More precisely,
if the operator ${\bf h}({\bf 0})$ has a zero energy resonance,
then the solution $\psi(\cdot)$ of $G\psi=\psi$ is equal to
$v_1(\cdot)$ (up to constant factor) and the vector $f=(f_0,
f_1),$ where
\begin{equation}\label{formula for f-alpha}
f_0={\rm const} \neq 0, \quad f_1(q)=-\frac{v_1(q) f_0}{\sqrt{2}
(l_1+l_2) \varepsilon(q)},
\end{equation}
obeys the equation ${\bf h}({\bf 0})f=0$ such that $f_1 \in
L_1({\Bbb T}^3) \setminus L_2({\Bbb T}^3)$  (see Lemma
\ref{solution of hf=0}). If the operator ${\bf h}({\bf 0})$ has a
zero eigenvalue, then the vector $f=(f_0, f_1),$ where $f_0$ and
$f_1$ are defined by \eqref{formula for f-alpha}, again obeys the
equation ${\bf h}({\bf 0})f=0$ and $f_1 \in L_2({\Bbb T}^3)$ (see
proof of the assertion (i) of Lemma \ref{resonance or
eigenvalue}).

As in the introduction, let us denote by $\tau_{\rm ess}(K)$ the
bottom of the essential spectrum of ${\bf H}(K)$ and by $N(K,z)$
the number of eigenvalues of ${\bf H}(K)$ on the left of $z,\,z
\leq \tau_{\rm ess}(K).$

Note that if the operator ${\bf h}({\bf 0})$ has either a zero
energy resonance or a zero eigenvalue, then for any $K \in
\Lambda$ and $p \in {\Bbb T}^3$ the operator ${\bf h}(K-p)+l_1
\varepsilon(p) {\bf I}$ is non-negative (see Lemma
\ref{positivity}), where ${\bf I}$ is the identity operator in
${\mathcal H}_0 \oplus {\mathcal H}_1.$ Hence Theorem \ref{THM
2.1} and equality $m_K=0,$ $K \in \Lambda$ imply that $\tau_{\rm
ess}(K)=0$ for all $K \in \Lambda.$

The main results of the present paper as follows.

\begin{thm}\label{THM 2.2} Let $K \in \Lambda$ and one of the following assumptions hold:\\
{\rm (i)} the operator ${\bf h}({\bf 0})$ has a zero eigenvalue;\\
{\rm (ii)} ${\bf h}({\bf 0}) \geq 0$ and a zero is the regular
type point for ${\bf h}({\bf 0}).$

Then the operator ${\bf H}(K)$ has finitely many negative
eigenvalues.
\end{thm}

\begin{thm}\label{THM 2.3} Let $K \in \Lambda.$ If the operator ${\bf h}({\bf 0})$ has a zero energy
resonance, then the operator ${\bf H}(K)$ has infinitely many
negative eigenvalues accumulating at zero and the function
$N(K,\cdot)$ obeys the relation
\begin{equation}\label{2.2}
\lim\limits_{z\to -0}\frac{N(K,z)}{|\log|z||}=\, {\mathcal U}_0,
\quad 0<{\mathcal U}_0<\infty.
\end{equation}
\end{thm}

\begin{rem}
The constant ${\mathcal U}_0$ does not depend on the function
$v_1(\cdot).$ It is positive and depends only on the ratio
$l_2/l_1.$
\end{rem}

\begin{rem}
Clearly, by equality \eqref{2.2} the infinite cardinality of the
negative discrete spectrum of ${\bf H}(K)$ follows automatically
from the positivity of ${\mathcal U}_0$.
\end{rem}

\begin{rem}
It is surprising that the asymptotics \eqref{2.2} doesn't depends
on the cardinality of $\Lambda,$ that is, this asymptotics is the
same for all $n \in {\Bbb N}.$ Since $\Lambda|_{n=1}=\{{\bf 0}\}$
in fact, a result similar to Theorem $\ref{THM 2.3}$ was proved in
\cite{ALR} for $n=1$ and $K=0.$
\end{rem}

\section{Some spectral properties of the family of Friedrichs models ${\bf h}(k)$}

In this section we study some spectral properties of the family of
Friedrichs models ${\bf h}(k),$ which plays an important role in
the study of spectral properties of ${\bf H}(K).$

Let the operator ${\bf h}_0(k),$ $k\in {\Bbb T}^3$ acts in
${\mathcal H}_0\oplus {\mathcal H}_1$ as
$$
{\bf h}_0(k)=\left( \begin{array}{cc}
0 & 0\\
0 & h_{11}(k)\\
\end{array}
\right).
$$

The perturbation ${\bf h}(k)-{\bf h}_0(k)$ of the operator ${\bf
h}_0(k)$ is a self-adjoint operator of rank 2, and thus, according
to the Weyl theorem, the essential spectrum of the operator ${\bf
h}(k)$ coincides with the essential spectrum of ${\bf h}_0(k).$ It
is evident that $\sigma_{\rm ess}({\bf h}_0(k))=[E_{\rm min}(k);
E_{\rm max}(k)],$ where the numbers $E_{\rm min}(k)$ and $E_{\rm
max}(k)$ are defined by
$$
E_{\rm min}(k):= \min_{q\in {\Bbb T}^3} E_k(q) \quad \mbox{and}
\quad E_{\rm max}(k):= \max_{q\in {\Bbb T}^3} E_k(q).
$$
This yields $\sigma_{\rm ess}({\bf h}(k))=[E_{\rm min}(k); E_{\rm
max}(k)].$

For any $k\in {\Bbb T}^3$ we define an analytic function
$\Delta(k\,; \cdot)$ (the Fredholm determinant associated with the
operator ${\bf h}(k)$) in ${\Bbb C} \setminus [E_{\rm min}(k);
E_{\rm max}(k)]$ by
$$
\Delta(k\,; z):=l_2
\varepsilon(k)+1-z-\frac{1}{2}\int_{{\Bbb T}^3}
\frac{v_1^2(t)dt}{E_k(t)-z}.
$$

The following lemma \cite{ALR, RDMFAT2019, RDNPCM2019} is a simple consequence of the
Birman-Schwinger principle and the Fredholm theorem.

\begin{lem}\label{LEM 3.1} For any $k\in {\Bbb T}^3$ the operator ${\bf h}(k)$ has an eigenvalue
$z \in {\Bbb C} \setminus [E_{\rm min}(k); E_{\rm max}(k)]$ if and
only if $\Delta(k\,; z)=0.$
\end{lem}

Since for any $k \in \Lambda$ the function $E_k(\cdot)$ has
non-degenerate zero minimum at the points of $\Lambda$ and the
function $v_1(\cdot)$ is a continuous on ${\Bbb T}^3,$ for any $k
\in {\Bbb T}^3$ the integral
$$
\int_{{\Bbb T}^3} \frac{v_1^2(t)dt}{E_k(t)}
$$
is positive and finite. The Lebesgue dominated convergence theorem
and the equality $\Delta({\bf 0}\,; 0)=\Delta(k\,; 0)$ for $k \in
\Lambda$ yield
$$
\Delta({\bf 0}\,; 0)=\lim\limits_{k\to k'} \Delta(k\,; 0),\quad
k'\in \Lambda.
$$

For some $\delta>0$ and $p_0 \in {\Bbb T}^3$ we set
$$
U_\delta(p_0):=\{p\in {\Bbb T}^3: |p-p_0|<\delta\}, \quad {\Bbb
T}_\delta: = {\Bbb T}^3 \setminus \bigcup_{q'\in \Lambda}
U_\delta(q').
$$

The following lemma establishes in which cases the bottom of the
essential spectrum is a threshold energy resonance or eigenvalue.

\begin{lem}\label{resonance or eigenvalue}
{\rm (i)} The operator ${\bf h}({\bf 0})$ has a zero eigenvalue if
and only if  $\Delta({\bf 0}\,; 0)=0$ and $v_1(q')=0$ for all
$q'\in \Lambda;$

{\rm (ii)} The operator ${\bf h}({\bf 0})$ has a zero energy
resonance if and only if $\Delta({\bf 0}\,; 0)=0$ and $v_1(q')
\neq 0$ for some $q'\in \Lambda.$
\end{lem}

\begin{proof} (i) "Only If Part". Suppose $f=(f_0, f_1)\in {\mathcal H}_0 \oplus {\mathcal H}_1$
is an eigenvector of the operator ${\bf h}({\bf 0})$ associated
with the zero eigenvalue. Then $f_0$ and $f_1$ satisfy the system
of equations
\begin{align}
& f_0+ \frac{1}{\sqrt{2}} \int_{{\Bbb T}^3} v_1(t)f_1(t)dt=0;\nonumber\\
& \frac{1}{\sqrt{2}} v_1(q)f_0+(l_1+l_2)\varepsilon(q)f_1(q)=0.\label{system of equations}
\end{align}

From \eqref{system of equations} we find that $f_0$ and $f_1$ are
given by \eqref{formula for f-alpha} and from the first equation
of \eqref{system of equations} we derive the equality $\Delta({\bf
0}\,; 0)=0.$

Now we show that $f_1 \in L_2({\Bbb T}^3)$ if and only if $v_1(q')
= 0$ for all $q' \in \Lambda.$ Indeed. If for some $q' \in
\Lambda$ we have $v_1(q')=0$ (resp. $v_1(q')\neq 0$), then there
exist the numbers $C_1,\,C_2,\,C_3>0,$ $\alpha \geq 1$ and
$\delta>0$ such that
\begin{equation}\label{estimate for v1}
C_1|q-q'|^\alpha \leq |v_1(q)| \leq C_2|q-q'|^\alpha,\quad q\in
U_\delta(q'),
\end{equation}
respectively
\begin{equation}\label{estimate for v1-1}
|v_1(q)|\geq C_3,\quad q\in U_\delta(q').
\end{equation}

The definition of the function $\varepsilon(\cdot)$ implies that
there exist the numbers $C_1,\,C_2,\,C_3>0$ and $\delta>0$ such
that
\begin{equation}\label{estimate for epsilon 1}
C_1|q-q'|^2 \leq \varepsilon(q) \leq C_2|q-q'|^2,\quad q\in
U_\delta(q'),\quad q' \in \Lambda
\end{equation}
\begin{equation}\label{estimate for epsilon 2}
\varepsilon(q) \geq C_3, \quad q \in {\Bbb T}_\delta.
\end{equation}

We have
\begin{equation}\label{integral for f1}
\int_{{\Bbb T}^3} |f_1(t)|^2dt=\frac{|f_0|^2}{2(l_1+l_2)^2}
\sum_{q' \in \Lambda}\, \int_{U_\delta(q')}
\frac{v_1^2(t)dt}{\varepsilon^2(t)}+ \frac{|f_0|^2}{2(l_1+l_2)^2}
\int_{{\Bbb T}_\delta} \frac{v_1^2(t)dt}{\varepsilon^2(t)}.
\end{equation}

If $v_1(q')=0$ for all $q' \in \Lambda,$ then using estimates
\eqref{estimate for v1}-\eqref{estimate for epsilon 2} we obtain
that
$$
\int_{{\Bbb T}^3} |f_1(t)|^2dt \leq C_1 \sum_{q' \in
\Lambda}\, \int_{U_\delta(q')}
\frac{|t-q'|^{2\alpha}}{|t-q'|^4}dt+C_2<\infty.
$$

In the case $v_1(q')\neq 0$ for some $q' \in \Lambda,$ an
application of estimates \eqref{estimate for v1-1},
\eqref{estimate for epsilon 1} imply
$$
\int_{{\Bbb T}^3} |f_1(t)|^2dt \geq C_1
\int_{U_\delta(q')} \frac{dt}{|t-q'|^4}=\infty.
$$

Therefore $f_1 \in L_2({\Bbb T}^3)$ if and only if $v_1(q') = 0$
for all $q' \in \Lambda.$

"If Part". Let $\Delta({\bf 0}\,; 0)=0$ and $v_1(q')=0$ for all
$q'\in \Lambda.$ Then the vector $f=(f_0, f_1),$ where $f_0$ and
$f_1$ are defined by \eqref{formula for f-alpha}, obeys the
equation ${\bf h}({\bf 0})f=0$ and as we show in "Only If Part"
that $f_1 \in L_2({\Bbb T}^3).$

(ii) "Only If Part". Let the operator ${\bf h}({\bf 0})$ have a
zero energy resonance. Then by Definition \ref{defn of resonance}
the equation
\begin{equation}\label{equation for resonance}
\psi(q)=\frac{v_1(q)}{2(l_1+l_2)}\int_{{\Bbb T}^3}
\frac{v_1(t)\psi(t)dt}{\varepsilon(t)}, \quad 
\psi\in C({\Bbb T}^3)
\end{equation}
has a simple solution $\psi\in C({\Bbb T}^3)$ and $\psi(q')\neq 0$
for some $q' \in \Lambda.$ It is easy to see that this solution is
equal to $v_1(\cdot)$ (up to a constant factor) and hence
$\Delta({\bf 0}\,; 0)=0.$

"If Part". Let the equality $\Delta({\bf 0}\,; 0)=0$ hold and
$v_1(q')\neq 0$ for some $q' \in \Lambda.$ Then the function
$v_1\in C({\Bbb T}^3)$ is a solution of the equation
\eqref{equation for resonance}, that is, the operator $h({\bf 0})$
has a zero energy resonance.
\end{proof}

Set
$$
\Lambda_0:=\{q' \in \Lambda: v_1(q') \neq 0\}.
$$

\begin{lem}\label{solution of hf=0} If the operator $h({\bf
0})$ has a zero energy resonance, then the vector $f=(f_0,f_1),$
where $f_0$ and $f_1$ are given by \eqref{formula for f-alpha},
obeys the equation $h({\bf 0})f=0$ and $f_1 \in L_1({\Bbb T}^3)
\setminus L_2({\Bbb T}^3).$
\end{lem}

\begin{proof}
Since the fact that the vector $f$ defined as in Lemma
\ref{solution of hf=0} satisfies $h({\bf 0})f=0$ is obvious, we
show that $f_1 \in L_1({\Bbb T}^3) \setminus L_2({\Bbb T}^3).$

Let the operator $h({\bf 0})$ have a zero energy resonance. Then
by the assertion (ii) of Lemma \ref{resonance or eigenvalue} we
have $v_1(q')\neq 0$ for some $q' \in \Lambda.$ Using the
estimates \eqref{estimate for v1}--\eqref{estimate for epsilon 2}
we have
\begin{align*}
\int_{{\Bbb T}^3} |f_1(t)|^2dt & \geq
\frac{|f_0|^2}{2(l_1+l_2)^2} \int_{U_\delta(q')}
\frac{v_1^2(t)dt}{\varepsilon^2(t)} \geq C_2 \int_{U_\delta(q')} \frac{dt}{|t-q'|^4}=\infty;\\
\int_{{\Bbb T}^3} |f_1(t)|dt &= \frac{
|f_0|}{\sqrt{2}(l_1+l_2)} \Bigl( \sum_{q' \in \Lambda_0}\,
\int_{U_\delta(q')} \frac{|v_1(t)|dt}{\varepsilon(t)}+
\sum_{q' \in \Lambda \setminus \Lambda_0}\,
\int_{U_\delta(q')} \frac{|v_1(t)|dt}{\varepsilon(t)}+
\int_{{\Bbb T}_\delta}
\frac{|v_1(t)|dt}{\varepsilon(t)} \Bigr) \\
& \leq C_1 \sum_{q' \in \Lambda_0}\, \int_{U_\delta(q')}
\frac{dt}{|t-q'|^2}+C_2 \sum_{q' \in \Lambda \setminus
\Lambda_0}\, \int_{U_\delta(q')}
\frac{dt}{|t-q'|^{2-\alpha}}+C_3<\infty.
\end{align*}
Therefore, $f_1 \in L_1({\Bbb T}^3) \setminus L_2({\Bbb T}^3).$
\end{proof}

\begin{lem}\label{positivity}
If the operator $h({\bf 0})$ has either a zero energy resonance or
a zero eigenvalue, then for any $K \in \Lambda$ and $p \in {\Bbb
T}^3$ the operator ${\bf h}(K-p)+l_1 \varepsilon(p) {\bf I}$ is
non-negative.
\end{lem}

Similar lemma were proved in \cite{ALR1} and we refer to this
paper for the proof.

Now we formulate a lemma (zero energy expansion for the Fredholm
determinant, leading to behaviors of the zero energy resonance),
which is important in the proof of Theorem \ref{THM 2.3}, that is,
the asymptotics \eqref{2.2}.

\begin{lem}\label{LEM 4.1} Let the operator $h({\bf 0})$ have a zero energy
resonance and $K, p' \in \Lambda.$ Then the following
decomposition
\begin{align*}
\Delta(K-p\,; z-l_1 \varepsilon(p))&=\frac{2
\pi^2}{n^2(l_1+l_2)^{3/2}} \Bigl( \sum_{q' \in \Lambda_0}
v_1^2(q') \Bigr) \sqrt{\frac{l_1^2+2 l_1 l_2}
{l_1+l_2}|p-p'|^2-\frac{2z}{n^2}}\\
& +O(|p-p'|^2)+O(|z|)
\end{align*}
holds for $|p-p'|\to 0$ and $z\to -0.$
\end{lem}

\begin{proof} Let us sketch the main idea of the proof.
Assume the operator $h({\bf 0})$ have a zero energy resonance and
$K, p' \in \Lambda.$ Using the additivity property of the integral
we represent the function $\Delta(K-p\,; z-l_1 \varepsilon(p))$ as
\begin{align}
\Delta(K-p\,; z-l_1 \varepsilon(p))&=w_1(K;p)-z-\frac{1}{2}
\sum\limits_{q' \in \Lambda_0}\, \int_{U_\delta(q')}
\frac{v_1^2(t)dt}{w_2(K; p,t)-z}\nonumber\\
&-\frac{1}{2} \sum\limits_{q' \in \Lambda \setminus \Lambda_0}\,
\int_{U_\delta(q')} \frac{v_1^2(t)dt}{w_2(K; p,t)-z}-\frac{1}{2} \int_{{\Bbb T}_\delta}
\frac{v_1^2(t)dt}{w_2(K; p,t)-z},\label{4.1}
\end{align}
where $\delta>0$ is a sufficiently small number.

Since the function $w_2(K; \cdot, \cdot)$ has non-degenerate
minimum at the points $(p', q'),$ $p',q' \in \Lambda,$ analysis
similar to \cite{ALR} show that
\begin{align*}
\int_{U_\delta(q')} \frac{v_1^2(t)dt}{w_2(K; p,t)-z}= &
\int_{U_\delta(q')} \frac{v_1^2(t)dt}{w_2(K; p',t)}-\frac{4 \pi^2 v_1^2(q')}{n^2 (l_1+l_2)^{3/2}}
\sqrt{\frac{l_1^2+2 l_1 l_2} {l_1+l_2}|p-p'|^2-\frac{2z}{n^2}}\\
& + O(|p-p'|^2)+O(|z|),\quad q' \in \Lambda_0;\\
\int_{U_\delta(q')} \frac{v_1^2(t)dt}{w_2(K; p,t)-z}= &
\int_{U_\delta(q')} \frac{v_1^2(t)dt}{w_2(K; p',t)}+O(|p-p'|^2)+O(|z|),\quad q' \in \Lambda \setminus \Lambda_0;\\
\int_{{\Bbb T}_\delta} \frac{v_1^2(t)dt}{w_2(K; p,t)-z}= &
\int_{{\Bbb T}_\delta} \frac{v_1^2(t)dt}{w_2(K; p',t)}+
O(|p-p'|^2)+O(|z|)
\end{align*}
as $|p-p'|\to 0$ and $z \to -0.$ Here we remind that $v_1(q')=0$
for all $q' \in \Lambda \setminus \Lambda_0$ and hence by the
estimate \eqref{estimate for v1} we have $v_1(q)=O(|q-q'|^\alpha)$
as $|q-q'| \to 0$ for some $\alpha \geq 1.$ Now substituting the
last three expressions and the the expansion
$$
w_1(K;p)=1+\frac{(l_1+l_2)n^2}{2} |p-p'|^2+O(|p-p'|^4)
$$
as $|p-p'|\to 0,$ to the equality \eqref{4.1} we obtain
\begin{align*}
\Delta(K-p\,; z-l_1 \varepsilon(p)) & =\Delta({\bf 0}\,; 0)+
\frac{2 \pi^2}{n^2(l_1+l_2)^{3/2}} \Bigl( \sum_{q' \in \Lambda_0}
v_1^2(q') \Bigr) \sqrt{\frac{l_1^2+2 l_1 l_2}
{l_1+l_2}|p-p'|^2-\frac{2z}{n^2}}\\
& +O(|p-p'|^2)+O(|z|)
\end{align*}
as $|p-p'|\to 0$ and $z\to -0.$ Since the operator $h({\bf 0})$
has a zero energy resonance by the assertion (ii) of Lemma
\ref{resonance or eigenvalue} we have the equality $\Delta({\bf
0}\,; 0)=0,$ which completes the proof of the Lemma \ref{LEM 4.1}.
\end{proof}

\begin{cor}\label{ineq for Delta1}
Let the operator $h({\bf 0},{\bf 0})$ have a zero energy resonance
and $K \in \Lambda.$ Then there exist the numbers $C_1, C_2,
C_3>0$ and $\delta>0$ such that\\
{\rm (i)} $C_1 |p-p'| \leq \Delta(K-p\,; -l_1 \varepsilon(p))
\leq C_2 |p-p'|,$ $p \in U_\delta(p'),$ $p' \in \Lambda;$\\
{\rm (ii)} $\Delta(K-p\,; -l_1 \varepsilon(p)) \geq C_3,$ $p \in
{\Bbb T}_\delta.$
\end{cor}

\begin{proof}
Lemma \ref{LEM 4.1} yields the assertion (i) for some positive
numbers $C_1, C_2.$ The positivity and continuity of the function
$\Delta(K-p\,; -l_1 \varepsilon(p))$ on the compact set ${\Bbb
T}_\delta$ imply the assertion (ii).
\end{proof}

\begin{lem}\label{estimate for Delta}
Let the operator $h({\bf 0})$ have a zero eigenvalue and $K \in
\Lambda.$ Then there exist the numbers $C_1, C_2,
C_3>0$ and $\delta>0$ such that\\
{\rm (i)} $C_1 |p-p'|^2 \leq \Delta(K-p\,; -l_1 \varepsilon(p))
\leq C_2 |p-p'|^2,$ $p \in U_\delta(p'),$ $p' \in \Lambda;$\\
{\rm (ii)} $\Delta(K-p\,; -l_1 \varepsilon(p)) \geq C_3,$ $p \in
{\Bbb T}_\delta.$
\end{lem}

\begin{proof} Let the operator $h({\bf 0})$
have a zero eigenvalue. Then by the assertion (i) of Lemma
\ref{resonance or eigenvalue} we have $v_1(p')=0$ for all $p' \in
\Lambda.$

Let $K \in \Lambda.$ Then $\Delta(K-p\,; -l_1
\varepsilon(p))=\Delta(p\,; -l_1 \varepsilon(p))$ holds for any $p
\in {\Bbb T}^3.$ Proceeding analogously to the proof of Lemma 3.4
of \cite{ALR1} one can show that the function $\Delta(\cdot\,;
-l_1 \varepsilon(\cdot))$ has minimum at the points $p=p'\in
\Lambda.$ Here we prove that this function has non-degenerate
minimum at the points $p=p' \in \Lambda.$ Since the function
$w_2({\bf 0}; \cdot, \cdot)$ is positive on $({\Bbb T}^3 \setminus
\Lambda) \times {\Bbb T}^3$ the integrals
$$
\lambda_{ij}^{(1)}(p):=\int_{{\Bbb T}^3} \left(
\frac{\partial^2 w_2({\bf 0}; p,t)}{\partial p^{(i)} \partial
p^{(j)}} \right) \frac{v_1^2(t)dt}{(w_2({\bf 0}; p,t))^2},\quad
i,j=1,2,3
$$
and
$$
\lambda_{ij}^{(2)}(p):=\int_{{\Bbb T}^3} \left(
\frac{\partial w_2({\bf 0}; p,t)}{\partial p^{(i)}} \frac{\partial
w_2({\bf 0}; p,t)}{\partial p^{(j)}} \right)
\frac{v_1^2(t)dt}{(w_2({\bf 0}; p,t))^3},\quad i,j=1,2,3
$$
are finite for any $p \in {\Bbb T}^3 \setminus \Lambda.$ The
condition $v_1(p')=0$ for all $p' \in \Lambda$ implies finiteness
of these integrals at the points of $\Lambda.$ Thus the functions
$\lambda_{ij}^{(l)}(\cdot),$ $l=1,2$ are continuous on ${\Bbb
T}^3.$

We define the function $I(\cdot)$ on ${\Bbb T}^3$ by
$$
I(p):=\int_{{\Bbb T}^3}\frac{v_1^2(t)dt}{w_2({\bf 0}; p,t)}.
$$

The function $I(\cdot)$ is a twice continuously differentiable
function ${\Bbb T}^3$ and
$$
\frac{\partial^2 I(p)}{\partial p^{(i)} \partial p^{(j)}}= -
\lambda_{ij}^{(1)}(p)+2\lambda_{ij}^{(2)}(p),\quad i,j=1,2,3.
$$

Simple calculations shows that
\begin{align*}
& \frac{\partial w_2({\bf 0}; p,q)}{\partial p^{(i)}}=n \left[ l_1
\sin(n q^{(i)})+l_2 \sin(n(p^{(i)}+q^{(i)})) \right], \quad
i=1,2,3;\\
& \frac{\partial^2 w_2({\bf 0}; p,q)}{\partial p^{(i)} \partial
p^{(i)}}=n^2\left[ l_1 \cos(n q^{(i)})+l_2
\cos(n(p^{(i)}+q^{(i)})) \right], \quad
i=1,2,3;\\
& \frac{\partial^2 w_2({\bf 0}; p,q)}{\partial p^{(i)} \partial
p^{(j)}}=0, \quad i \neq j, \quad i,j=1,2,3
\end{align*}
and hence for $p' \in \Lambda$ we obtain
\begin{align*}
& \frac{\partial^2 I(p')}{\partial p^{(i)} \partial
p^{(i)}}=-\frac{n^2}{4} \int_{{\Bbb T}^3} \left(
\sum\limits_{{l=1,\,l \neq i}}^3 (1-\cos(n t^{(l)})) \right)
\frac{(1+\cos(n t^{(i)})) v_1^2(t)}{\varepsilon^3(t)}dt, \,\,
i=1,2,3;\\
& \frac{\partial^2 I(p')}{\partial p^{(i)} \partial
p^{(j)}}=\frac{n^2}{4} \int_{{\Bbb T}^3} \frac{\sin(n
t^{(i)}) \sin(n t^{(j)}) v_1^2(t)}{\varepsilon^3(t)}dt, \quad i
\neq j, \quad i,j=1,2,3.
\end{align*}

The last equalities and the evenness of $v_2^2(\cdot)$ on each
variables imply
$$
\frac{\partial^2 I(p')}{\partial p^{(i)} \partial p^{(i)}}<0,
\quad \frac{\partial^2 I(p')}{\partial p^{(i)}
\partial p^{(j)}}=0, \quad i \neq j, \quad
i,j=1,2,3
$$
for $p' \in \Lambda.$ Since
$$
\frac{\partial^2 w_1({\bf 0};p')}{\partial p^{(i)} \partial
p^{(i)}}=(l_1+l_2)n^2, \quad \frac{\partial^2 w_1({\bf 0};p')}
{\partial p^{(i)}\partial p^{(j)}}=0, \quad i \neq j, \quad
i,j=1,2,3
$$
for all $p' \in \Lambda$ by definition of $\Delta(\cdot\,; \cdot)$
we have
$$
\frac{\partial^2 \Delta(p'\,; 0)}{\partial p^{(i)} \partial
p^{(i)}}>(l_1+l_2)n^2, \quad \frac{\partial^2 \Delta(p'\,;
0)}{\partial p^{(i)} \partial p^{(j)}}=0, \quad i \neq j, \quad
i,j=1,2,3
$$
for all $p' \in \Lambda.$ Therefore the function $\Delta(\cdot\,;
-l_1 \varepsilon(\cdot))$ has non-degenerate minimum at the points
of $p=p' \in \Lambda.$ This fact completes the proof of lemma.
\end{proof}

\begin{lem}\label{estimate-regular point} Let $K \in \Lambda$ and zero be the regular type point for ${\bf h}({\bf 0})$
with ${\bf h}({\bf 0}) \geq 0.$ Then there exists a positive
number $C_1$ such that the inequality
$$
\Delta(K-p\,; z-l_1 \varepsilon(p)) \geq C_1
$$
holds for any $p \in {\Bbb T}^3$ and $z<0.$
\end{lem}

\begin{proof}
Let zero be the regular type point of ${\bf h}({\bf 0}),$ that is,
$\Delta({\bf 0}\,; 0) \not=0.$ Assume $\Delta({\bf 0}\,; 0)<0.$
Then $\lim\limits_{z\to -\infty}\Delta({\bf 0}\,; z)=-\infty$ and
the continuity of of the function $\Delta({\bf 0}\,; \dot)$ on
$(-\infty; 0]$ imply that there exists $z_0<0$ such that
$\Delta({\bf 0}\,; z_0)=0.$ In this case by Lemma \ref{LEM 3.1}
the number $z_0$ is an eigenvalue of the operator $h({\bf 0}).$ On
the other hand by the assumption of the lemma we have ${\bf
h}({\bf 0}) \geq 0.$ Therefore the operator $h({\bf 0})$ has no
negative eigenvalues. This contrary gives $\Delta({\bf 0}\,;
0)>0.$

Since for any $K \in \Lambda$ the function $\Delta(K-p\,; -l_1
\varepsilon(p))$ has minimum at the points $p=p'\in \Lambda,$ for
all $p \in {\Bbb T}^3$ and $z<0$ we obtain
$$
\Delta(K-p\,; z-l_1 \varepsilon(p)) > \Delta(K-p\,; -l_1
\varepsilon(p)) \geq \Delta({\bf 0}\,; 0)>0.
$$
Setting $C_1:=\Delta({\bf 0}\,; 0)$ we complete the proof of
lemma.
\end{proof}

\section{The Birman-Schwinger principle.}

For a bounded self-adjoint operator $A$ acting in the Hilbert
space ${\mathcal R},$ we define the number $n(\gamma, A)$ by the
rule
$$
n(\gamma, A):=\sup \{{\rm dim} F: (Au,u)>\gamma,\, u\in F \subset
{\mathcal R},\,||u||=1\}.
$$

The number $n(\gamma, A)$ is equal to the infinity if
$\gamma<\max\sigma_{\rm ess}(A);$ if $n(\gamma, A)$ is finite,
then it is equal to the number of the eigenvalues of $A$ bigger
than $\gamma.$

By the definition of $N(K,z),$ we have
$$
N(K,z)=n(-z, -{\bf H}(K)),\,-z>-\tau_{\rm ess}(K).
$$

Since for any $K \in {\Bbb T}^3$ the function $\Delta(K-p\,; z-l_1
\varepsilon(p))$ is a positive on $(p,z) \in {\Bbb T}^3 \times
(-\infty; \tau_{\rm ess}(K)),$ the positive square root of
$\Delta(K-p\,; z-l_1 \varepsilon(p))$ exists for any $K,p\in {\Bbb
T}^3$ and $z < \tau_{\rm ess}(K).$

In our analysis of the discrete spectrum of ${\bf H}(K),$ $K \in
{\Bbb T}^3$ the crucial role is played by the self-adjoint compact
$2 \times 2$ block operator matrix $\widehat{{\bf T}}(K,z),\,z <
\tau_{\rm ess}(K)$ acting on ${\mathcal H}_0 \oplus {\mathcal
H}_1$ as
$$
\widehat{{\bf T}}(K,z):=\left( \begin{array}{cc}
\widehat{T}_{00}(K,z) & \widehat{T}_{01}(K,z)\\
\widehat{T}_{01}^*(K,z) & \widehat{T}_{11}(K,z)\\
\end{array}
\right)
$$
with the entries
\begin{align*}
& \widehat{T}_{00}(K,z)g_0=(1+z-w_0(K))g_0,\quad
\widehat{T}_{01}(K,z)g_1= -\int_{{\Bbb T}^3}
\frac{v_0(t)g_1(t)dt}{\sqrt{\Delta(K-t\,; z-l_1 \varepsilon(t))}};\\
& (\widehat{T}_{11}(K,z)g_1)(p)=\frac{v_1(p)}{2\sqrt{\Delta(K-p\,;
z-l_1 \varepsilon(p))}} \int_{{\Bbb T}^3}
\frac{v_1(t)g_1(t)dt}{\sqrt{\Delta(K-t\,; z-l_1
\varepsilon(t))}(w_2(K; p,t)-z)}.
\end{align*}

The following lemma is a modification of the well-known
Birman-Schwinger principle for the operator ${\bf H}(K)$ (see
\cite{AL, ALM, ALR, Sob}).

\begin{lem}\label{LEM 4.3}
Let $K \in {\Bbb T}^3.$ The operator $\widehat{{\bf T}}(K,z)$ is
compact and continuous in $z < \tau_{\rm ess}(K)$ and
$$
N(K,z) = n(1, \widehat{{\bf T}}(K,z)).
$$
\end{lem}

For the proof of this lemma, see Lemma 5.1 of \cite{ALR}.

\section{Finiteness of the number of eigenvalues of ${\bf H}(K),$ $K \in \Lambda$}

We starts the proof of the finiteness of the number of negative
eigenvalues (Theorem \ref{THM 2.2}) with the following two lemmas.

\begin{lem}\label{estimate for w2}
Let $K,p',q' \in \Lambda.$ Then there exist the numbers $C_1,
C_2>0$ and $\delta>0$ such that\\
{\rm (i)} $C_1(|p-p'|^2+|q-q'|^2)\leq w_2(K;p,q) \leq
C_2(|p-p'|^2+|q-q'|^2),$ $(p,q)\in U_\delta(p') \times
U_\delta(q');$\\
{\rm (ii)} $w_2(K;p,q) \geq C_1,$ $(p,q) \not\in
\bigcup\limits_{p'\in \Lambda} U_\delta(p') \times
\bigcup\limits_{q'\in \Lambda} U_\delta(q').$
\end{lem}

\begin{proof}
Since for any $K \in \Lambda$ the function $w_2(K;\cdot,\cdot)$
has non-degenerate zero minimum at the points $(p', q') \in
\Lambda \times \Lambda,$ we obtain the following expansion
\begin{align*}
w_2(K;p,q)=&\frac{n^2}{2} \left[ (l_1+l_2)|p-p'|^2+2l_2(p-p',
q-q')+(l_1+l_2)|q-q'|^2 \right]\\
& +O(|p-p'|^4)+O(|q-q'|^4)
\end{align*}
as $|p-p'|,\,|q-q'|\to 0$ for $p',q'\in \Lambda.$ Then  there
exist positive numbers $C_1, C_2$ and $\delta$ so that (i) and
(ii) hold true.
\end{proof}

\begin{lem}\label{T(K,z) compact}
Let $K \in \Lambda$ and one of the following assumptions hold:\\
{\rm (i)} the operator ${\bf h}({\bf 0})$ has a zero eigenvalue;\\
{\rm (ii)} a zero is the regular type point for ${\bf h}({\bf 0})$
and ${\bf h}({\bf 0}) \geq 0.$

Then for any $z \leq 0$ the operator $\widehat{{\bf T}}(K,z)$ is
compact and continuous from the left up to $z=0.$
\end{lem}

\begin{proof}
Let $K \in \Lambda.$ Denote by $Q(K;p,q;z)$ the kernel of the
integral operator $\widehat{T}_{11}(K,z),$ $z<0,$ that is,
$$
Q(K;p,q;z):=\frac{v_1(p)v_1(q)}{2\sqrt{\Delta(K-p\,; z-l_1
\varepsilon(p))}(w_2(K;p,q)-z) \sqrt{\Delta(K-q\,; z-l_1
\varepsilon(q))}}.
$$

If the operator ${\bf h}({\bf 0})$ has a zero eigenvalue, then by
the assertion (i) of Lemma \ref{resonance or eigenvalue} we have
$v_1(q')=0$ for all $q' \in \Lambda.$ By  virtue of inequality
\eqref{estimate for v1}, Corollary \ref{ineq for Delta1} and Lemma
\ref{estimate for w2} the kernel $Q(K;p,q;z)$ is estimated by
\begin{equation}\label{estimate function1}
C_1 \sum_{p',q' \in \Lambda} \left(
\frac{\chi_{\delta}(p-p')}{|p-p'|}+1 \right)\left( \frac{|q-q'|
\chi_\delta(p-p')\chi_\delta(q-q')}{|p-p'|^2+|q-q'|^2}+1\right)\left(
\frac{\chi_{\delta}(q-q')}{|q-q'|^{\frac{1}{2}}}+1\right),
\end{equation}
where $\chi_\delta(\cdot)$ is the characteristic function of
$U_\delta({\bf 0}).$

If a zero is the regular type point for ${\bf h}({\bf 0})$ and
${\bf h}({\bf 0}) \geq 0,$ then by virtue of Lemmas
\ref{estimate-regular point} and \ref{estimate for w2} the kernel
$Q(K;p,q;z)$ is estimated by
\begin{equation}\label{estimate function2}
C_1 \sum_{p',q' \in \Lambda} \left(
\frac{\chi_\delta(p-p')\chi_\delta(q-q')}{|p-p'|^2+|q-q'|^2}+1\right).
\end{equation}

The functions \eqref{estimate function1} and \eqref{estimate
function2} are square-integrable on $({\Bbb T}^3)^2$ and hence for
any $z\leq 0$ the operator $\widehat{T}_{11}(K,z)$ is
Hilbert-Schmidt.

The kernel function of $\widehat{T}_{11}(K,z),$ $z<0$ is
continuous in $p,q\in {\Bbb T}^3.$ Therefore the continuity of the
operator $\widehat{T}_{11}(K,z)$ from the left up to $z=0$ follows
from Lebesgue's dominated convergence theorem.

Since for all $z \leq 0$ the operators $\widehat{T}_{00}(K,z),$
$\widehat{T}_{01}(K,z)$ and $\widehat{T}_{01}^*(K,z)$ are of rank
1 and continuous from the left up to $z=0$ one concludes that
$\widehat{\bf T}(K,z)$ is compact and continuous from the left up
to $z=0.$
\end{proof}

We are now ready for the

\begin{proof}[Proof of Theorem $\ref{THM 2.2}$]
Let the conditions of Theorem \ref{THM 2.2} be fulfilled. Using
the Weyl inequality
$$
n(\lambda_1+\lambda_2, A_1+A_2) \leq n(\lambda_1,
A_1)+n(\lambda_2, A_2)
$$
for the sum of compact operators $A_1$ and $A_2$ and for any
positive numbers $\lambda_1$ and $\lambda_2$ we have
\begin{equation}\label{Weyl inequality}
n(1, \widehat{\bf T}(K,z))\leq n(1/2, \widehat{\bf T}(K,0))+n(1/2,
\widehat{\bf T}(K,z)-\widehat{\bf T}(K,0))
\end{equation}
for all $z<0.$

By virtue of Lemma \ref{T(K,z) compact} the operator $\widehat{\bf
T}(K,z)$ is continuous from the left up to $z=0,$ which implies
that the second summand on the r.h.s. of \eqref{Weyl inequality}
tends to zero as $z \to -0.$ By Lemma \ref{LEM 4.3} we have
$N(K,z)=n(1, \widehat{\bf T}(K,z))$ as $z<0$ and hence
$$
\lim\limits_{z\to -0}N(K,z)=N(K,0)\leq n(1/2, \widehat{\bf
T}(K,0)).
$$
Thus $N(K,0)\leq n(1/2, \widehat{\bf T}(K,0)).$ By Lemma
\ref{T(K,z) compact} the number $n(1/2, \widehat{\bf T}(K,0))$ is
finite and hence $N(K,0)<\infty.$ This completes the proof of
Theorem \ref{THM 2.2}.
\end{proof}

\section{Asymptotics for the number of negative eigenvalues of ${\bf H}(K),$ $K \in \Lambda$}

In this section first we derive the asymptotic relation
\eqref{2.2} for the number of negative eigenvalues of ${\bf
H}(K),$ $K \in \Lambda.$

Let ${\Bbb S}^2$ be the unit sphere in ${\Bbb R}^3$ and ${\bf
\sigma}=L_2({\Bbb S}^2).$ As we shall see, the discrete spectrum
asymptotics of the operator $\widehat{{\bf T}}(K,z)$ $K \in
\Lambda$ as $z \to -0$ is determined by the integral operator
$S_{\bf r},$ ${\bf r}=1/2 |\log|z||$ in $L_2((0, {\bf r}), {\bf
\sigma})$ with the kernel
$$
S(y,t):=\frac{1}{4 \pi^2} \frac{(l_1+l_2)^2}{\sqrt{l_1^2+2l_1l_2}}
\frac{1}{(l_1+l_2) \cosh y+l_2 t},
$$
where $y=x-x',$ $x, x' \in (0, {\bf r})$ and $t=\langle \xi, \eta
\rangle$ is the inner product of the arguments $\xi, \eta \in
{\Bbb S}^2.$

The eigenvalues asymptotics for the operator $S_{\bf r}$ have been
studied in detail by Sobolev \cite{Sob}, by employing an argument
used in the calculation of the canonical distribution of Toeplitz
operators.

Let us recall some results of \cite{Sob} which are important in
our work.

The coefficient in the asymptotics \eqref{2.2} of $N(K,z)$ will be
expressed by means of the self-adjoint integral operator
$\widehat{S}(\theta),$ $\theta \in {\Bbb R},$ in the space
$\sigma,$ whose kernel is of the form
$$
\widehat{S}(\theta, t):=\frac{1}{4 \pi^2}
\frac{(l_1+l_2)^2}{l_1^2+2l_1l_2} \frac{\sinh [\theta \arccos
\frac{l_2}{l_1+l_2} t]}{\sinh (\pi \theta)},
$$
and depends on $t=\langle \xi, \eta \rangle.$ For $\gamma>0,$
define
$$
U(\gamma):=\frac{1}{4 \pi} \int_{-\infty}^{+\infty}
n(\gamma, \widehat{S}(\theta)) d \theta.
$$
This function was studied in detail in \cite{Sob}; where it was
used in showing existence proof of the Efimov effect. In
particular, as it was shown in \cite{Sob}, the function $U(\cdot)$
is continuous in $\gamma>0$, and the limit
\begin{equation}\label{U gamma}
\lim\limits_{{\bf r} \to 0} \frac{1}{2} {\bf r}^{-1} n(\gamma,
S_{\bf r})=U(\gamma)
\end{equation}
exists and the number $U(1)$ is positive.

For completeness, we reproduce the following lemma, which has been
proven in \cite{Sob}.

\begin{lem}\label{LEM 4.4} Let $A(z)= A_0(z)+A_1(z),$ where $A_0(z)$
$(A_1(z))$ is compact and continuous for $z < 0$ $($for $z\leq
0).$ Assume that the limit
$$
\lim\limits_{z\to -0}f(z)\,n(\gamma, A_0(z)) = l(\gamma)
$$
exists and $l(\cdot)$ is continuous in $(0; +\infty)$ for some
function $f(\cdot),$ where $f(z)\to 0$ as $z \to -0.$ Then the
same limit exists for $A(z)$ and
$$
\lim\limits_{z\to -0}f(z)\,n(\gamma, A(z)) = l(\gamma).
$$
\end{lem}

\begin{rem}
Since the function $U(\cdot)$ is continuous with respect to
$\gamma,$ it follows from Lemma $\ref{LEM 4.4}$ that any
perturbation of $A_0(z)$ treated in Lemma $\ref{LEM 4.4}$ $($which
is compact and continuous up to $z=0)$  does not contribute to the
asymptotic relation $\eqref{2.2}.$ In the rest part of this
subsection we use this fact without further comments.
\end{rem}

Now we are going to reduce the study of the asymptotics for the
operator $\widehat{{\bf T}}(K,z)$ with $K \in \Lambda$ to that of
the asymptotics $S_{\bf r}.$

Let ${\bf T}(\delta; |z|)$ be the operator in ${\mathcal H}_0
\oplus {\mathcal H}_1$ defined by
$$
{\bf T}(\delta; |z|):=\left( \begin{array}{cc}
0 & 0\\
0 & T_{11}(\delta; |z|)\\
\end{array}
\right),
$$
where $T_{11}(\delta; |z|)$ is the integral operator in ${\mathcal
H}_1$ with the kernel
\begin{align*}
D \sum_{p',q' \in \Lambda_0}
\frac{v_1(p')v_1(q')\chi_\delta(p-p')\chi_\delta(q-q')(m |p-p'|^2+
2|z|/n^2)^{-\frac{1}{4}}(m |q-p'|^2+2|z|/n^2)^{-\frac{1}{4}}}
{(l_1+l_2)|p-p'|^2+2l_2(p-p', q-q')+(l_1+l_2)|q-q'|^2+2|z|/n^2}.
\end{align*}
Here
$$
D:=\frac{(l_1+l_2)^{3/2}}{\pi^2} \Bigl( \sum_{q' \in \Lambda_0}
v_1^2(q') \Bigr)^{-1} \quad \mbox{and} \quad
m:=\frac{l_1^2+2l_1l_2}{l_1+l_2}.
$$

The operator ${\bf T}(\delta; |z|)$ is called singular part of
$\widehat{{\bf T}}(K,z).$

The main technical point to apply Lemma \ref{LEM 4.4} is the
following lemma.

\begin{lem}\label{LEM 4.5}
Let $K \in \Lambda.$ Then for any $z\leq 0$ and small $\delta>0$
the difference $\widehat{{\bf T}}(K,z)-{\bf T}(\delta; |z|)$ is
compact and is continuous with respect to $z \leq 0.$
\end{lem}

\begin{proof}
By Lemma \ref{estimate for w2} and Corollary \ref{ineq for
Delta1}, one can estimate the kernel of the operator
$\widehat{T}_{11}(K,z)-T_{11}(\delta; |z|),$ $z \leq 0,$ by the
square-integrable function
\begin{align*}
C_1 \sum_{p',q' \in \Lambda_0} \Bigl[ & \frac{1}{|p-p'|^{1/2}}+
\frac{1}{|q-q'|^{1/2}}+
\frac{|p-p'|+|q-q'|}{|p-p'|^{1/2}(|p-p'|^2+|q-q'|^2)|q-q'|^{1/2}}\\
& +\frac{|z|^{1/2}} {(|p-p'|^2+|z|)^{1/4} (|p-p'|^2+|q-q'|^2)
(|q-q'|^2+|z|)^{1/4}}+1 \Bigr].
\end{align*}
Hence, the operator $\widehat{T}_{11}(K,z)-T_{11}(\delta; |z|)$
belongs to the Hilbert-Schmidt class for all $z\leq 0.$ In
combination with the continuity of the kernel of the operator with
respect to $z<0,$ this implies the continuity of
$\widehat{T}_{11}(K,z)-T_{11}(\delta; |z|)$ with respect to $z
\leq 0.$

It is easy to see that $\widehat{T}_{00}(K,z),$
$\widehat{T}_{01}(K,z)$ and $\widehat{T}_{01}^*(K,z)$ are rank 1
operators and they are continuous from the left up to $z=0.$
Consequently $\widehat{{\bf T}}(K,z)-{\bf T}(\delta; |z|)$ is
compact and continuous in $z \leq 0.$
\end{proof}

From definition of ${\bf T}(\delta; |z|)$ it follows that
$\sigma({\bf T}(\delta; |z|))=\{0\} \cup \sigma(T_{11}(\delta;
|z|))$ and hence $n(\gamma, {\bf T}(\delta; |z|))=n(\gamma,
T_{11}(\delta; |z|))$ for all $\gamma>0.$

The following theorem is fundamental for the proof of the
asymptotic relation \eqref{2.2}.

\begin{thm}\label{THM 4.6}
We have the relation
\begin{equation}\label{4.2}
\lim\limits_{|z| \to 0} \frac{n(\gamma, T_{11}(\delta;
|z|))}{|\log |z|| }=U(\gamma), \quad \gamma>0.
\end{equation}
\end{thm}

\begin{proof} The subspace of functions $\psi,$ supported by the
set $\bigcup\limits_{q'\in \Lambda_0} U_\delta(q')$ is invariant
with respect to the operator $T_{11}(\delta; |z|).$ Let
$T_{11}^0(\delta; |z|)$ be the restriction of the integral
operator $T_{11}(\delta; |z|)$ to the subspace
$L_2(\bigcup\limits_{q'\in \Lambda_0} U_\delta(q')),$ that is, the
integral operator in $L_2(\bigcup\limits_{q'\in \Lambda_0}
U_\delta(q'))$ with the kernel $T_{11}^0(\delta; |z|;
\cdot,\cdot)$ defined on $\bigcup\limits_{p'\in \Lambda_0}
U_\delta(p') \times \bigcup\limits_{q'\in \Lambda_0} U_\delta(q')$
as
\begin{align*}
T_{11}^0(\delta; |z|; p,q):=\frac{D (m |p-p'|^2+
2|z|/(n^2))^{-\frac{1}{4}}(m |q-q'|^2+2|z|/(n^2))^{-\frac{1}{4}}}
{(l_1+l_2)|p-p'|^2+2l_2(p-p', q-q')+(l_1+l_2)|q-q'|^2+2|z|/(n^2)},
\end{align*}
$(p,q) \in U_\delta(p') \times U_\delta(q')$ for $p',q' \in
\Lambda_0.$

In the rest part of the proof we denote by $n_0$ the number of
points of $\Lambda_0$ and for convenience we numerate the points
of $\Lambda_0$ as $p_1, \ldots, p_{n_0}$ and set
$\overline{1,n_0}=1,\ldots,n_0.$

Since $L_2(\bigcup\limits_{q'\in \Lambda_0} U_\delta(q')) \cong
\bigoplus\limits_{q' \in \Lambda_0} L_2(U_\delta(q')),$ we can
express the integral operator $T_{11}^0(\delta; |z|)$ as the $n_0
\times n_0$ block operator matrix ${\bf T}_0(\delta; |z|)$ acting
on $\bigoplus\limits_{i=1}^{n_0} L_2(U_\delta(p_{i}))$ as
$$
{\bf T}_0(\delta; |z|):=\left( \begin{array}{ccc}
T_0^{(1,1)}(\delta; |z|) & \ldots & T_0^{(1,n_0)}(\delta; |z|)\\
\vdots & \ddots & \vdots\\
T_0^{(n_0,1)}(\delta; |z|) & \ldots & T_0^{(n_0,n_0)}(\delta; |z|)\\
\end{array}
\right),
$$
where for $i,j=\overline{1, n_0}$ the operator
$T_0^{(i,j)}(\delta; |z|): L_2(U_\delta(p_{j})) \to
L_2(U_\delta(p_{i}))$ is the integral operator with the kernel
$T_0(\delta; |z|; p,q),$ $(p,q) \in U_\delta(p_{i}) \times
U_\delta(p_{j}).$

Set
$$
L_2^{(n_0)}(U_r({\bf 0})):=\{\phi=(\phi_1, \cdots, \phi_{n_0}):\,
\phi_i \in L_2(U_r({\bf 0})),\, i=\overline{1, n_0} \}.
$$

It is easy to show that ${\bf T}_0(\delta; |z|)$ is unitarily
equivalent to the $n_0 \times n_0$ block operator matrix ${\bf
T}_1(r),$ $r=|z|^{-\frac{1}{2}},$ acting on $ L_2^{(n_0)}(U_r({\bf
0}))$ as
$$
{\bf T}_1(r):=\left( \begin{array}{ccc}
v_1(p_1) v_1(p_1) T_1(r) & \ldots & v_1(p_1) v_1(p_{n_0}) T_{1}(r)\\
\vdots & \ddots & \vdots\\
v_1(p_{n_0}) v_1(p_1) T_{1}(r) & \ldots & v_1(p_{n_0})v_1(p_{n_0}) T_{1}(r)\\
\end{array}
\right),
$$
where $T_{1}(r)$ is the integral operator on $L_2(U_r({\bf 0}))$
with the kernel
$$
\frac{D (m |p|^2+ 2/(n^2))^{-\frac{1}{4}}(m
|q|^2+2/(n^2))^{-\frac{1}{4}}} {(l_1+l_2)|p|^2+2l_2(p,
q)+(l_1+l_2)|q|^2+2/(n^2)}.
$$

The equivalence is realized by the unitary dilation ($n_0 \times
n_0$ diagonal matrix)
$$
{\bf B}_r:={\rm diag} \{ B_r^{(1)}, \ldots, B_r^{(n_0)}\}:
\bigoplus\limits_{i=1}^{n_0} L_2(U_\delta(p_{i})) \to
L_2^{(n_0)}(U_r({\bf 0})),
$$
Here for $i=\overline{1, n_0}$ the operator $B_r^{(i)}:
L_2(U_\delta(p_i)) \to L_2(U_r({\bf 0}))$ acts as
$$
(B_r^{(i)}f)(p)=\left( r/\delta \right)^{-3/2} f(\delta p/r+p_i).
$$

Let ${\bf A}_r$ and ${\bf E}$ be the $n_0 \times 1$ and $1 \times
n_0$ matrices of the form
$$
{\bf A}_r:=\left( \begin{array}{ccc}
v_1(p_1) T_{1}(r)\\
\vdots\\
v_1(p_{n_0}) T_{1}(r)\\
\end{array}
\right), \quad {\bf E}:=(v_1(p_1) I \ldots v_1(p_{n_0}) I),
$$
respectively, where $I$ is the identity operator on $L_2(U_r({\bf
0})).$

It is well known that if $B_1, B_2$ are bounded operators and
$\gamma \neq 0$ is an eigenvalue of $B_1B_2,$ then $\gamma$ is an
eigenvalue for $B_2B_1$ as well of the same algebraic and
geometric multiplicities (see {\it e.g.} \cite{Hal}). Therefore,
$n(\gamma, {\bf A}_r {\bf E})=n(\gamma, {\bf E} {\bf A}_r),$
$\gamma>0.$ Direct calculation shows that ${\bf T}_1(r)={\bf A}_r
{\bf E}$ and
$$
{\bf E}{\bf A}_r= T_1^0(r):=\Bigl( \sum\limits_{i=1}^{n_0}
v_1^2(p_i) \Bigr) T_1(r).
$$
So, for $\gamma>0$ we have $n(\gamma, {\bf T}_1(r))=n(\gamma,
T_1^0(r)).$

Furthermore, replacing
$$
(m|p|^2+2/(n^2))^{\frac{1}{4}}, \quad
(m|q|^2+2/(n^2))^{\frac{1}{4}} \quad \mbox{and} \quad
(l_1+l_2)|p|^2+2l_2(p, q)+(l_1+l_2)|q|^2+2/(n^2)
$$
by the expressions
$$
(m|p|^2)^{\frac{1}{4}}(1-\chi_1(p))^{-1}, \quad
(m|q|^2)^{\frac{1}{4}}(1-\chi_1(q))^{-1}\quad \mbox{and}\quad
(l_1+l_2)|p|^2+2l_2(p, q)+(l_1+l_2)|q|^2,
$$
respectively, we obtain the integral operator $T_2(r).$ The error
$T_1^0(r)-T_2(r)$ is a Hilbert-Schmidt operator and continuous up
to $z=0.$

Using the dilation
$$
M: L_2(U_r({\bf 0})\setminus U_1({\bf 0})) \to L_2((0, {\bf r}),
{\bf \sigma}), \quad (Mf)(x,w)=e^{3x/2}f(e^x w),
$$
where ${\bf r} = 1/2 |\log |z||,$ $x \in (0, {\bf r}),$ $w \in
{\Bbb S}^2,$ one sees that the operator $T_2(r)$ is unitarily
equivalent to the integral operator $S_{\bf r}.$

Since the difference of the operators $S_{\bf r}$ and
$T_{11}(\delta; |z|)$ is compact (up to unitary equivalence) and
hence, since ${\bf r} = 1/2 |\log |z||,$ we obtain the equality
\begin{equation*}
\lim\limits_{|z| \to 0} \frac{n(\gamma, T_{11}(\delta;
|z|))}{|\log |z||}=\lim\limits_{{\bf r} \to 0} \frac{1}{2} {\bf
r}^{-1} n(\gamma, S_{\bf r}), \quad \gamma>0.
\end{equation*}
Now Lemma \ref{LEM 4.4} and the equality \eqref{U gamma} complete
the proof of Theorem \ref{THM 4.6}.
\end{proof}

We are now ready for the

\begin{proof}[Proof of Theorem $\ref{THM 2.3}$] Let the operator ${\bf h}({\bf 0})$
have a zero energy resonance and $K \in \Lambda.$ Using Lemmas
\ref{LEM 4.4}, \ref{LEM 4.5} and Theorem \ref{THM 4.6} we have
that
$$
\lim\limits_{|z| \to 0} \frac{n(1, {\bf T}(K,z))}{|\log |z|| }=
U(1).
$$
Taking into account the last equality and Lemma \ref{LEM 4.3}, and
setting ${\mathcal U}_0=U(1),$ we complete the proof of Theorem
\ref{THM 2.3}.
\end{proof}

{\bf Acknowledgements.} This work was supported in part by the Malaysian
Ministry of Education through the Research Management Centre (RMC),
Universiti Tekhnology Malaysia (PAS, Ref. No. PY/2014/04068, Vote:
QJ130000.2726.01K82).

\end{document}